\documentclass[copyright,creativecommons]{eptcs}

\providecommand{\event}{WWV 2013} 
\usepackage{breakurl}             

\usepackage{latexsym,amsmath,amssymb}
\usepackage{amsthm}
\usepackage{stmaryrd}
\usepackage{extarrows}

\newtheorem{definition}{Definition}
\newtheorem{lemma}{Lemma}
\newtheorem{proposition}{Proposition}
\newtheorem{theorem}{Theorem}
\newtheorem{example}{Example}

\newcommand{\p}[1]{\mathtt{#1}}

\newcommand{\trans}[3]{\p #2 \rightarrow \p #3 \!:\! #1}
\newcommand{\transl}[3]{\p #2 \rightarrow \p #3 : #1 } 
\newcommand{\op}[1]{\mathsf{#1}}
\newcommand{\role}[2]{[\,#1\,]_{\p #2}}
\newcommand{\co}[1]{\overline{#1}}
\newcommand{\msg}[1]{\langle #1 \rangle}
\newcommand{\one}{\mathbf{1}}
\newcommand{\zero}{\mathbf{0}}
\newcommand{\tick}{\surd}
\newcommand{\arro}[1]{\xrightarrow[]{#1}}
\DeclareMathOperator{\transI}{transI}
\DeclareMathOperator{\transF}{transF}
\DeclareMathOperator{\roles}{roles}

\DeclareMathOperator{\fconf}{\mbox{}^f\#}
\DeclareMathOperator{\proj}{proj}

\def \mathrule #1#2#3{\begin{array}{l}
    {\mbox{\scriptsize ({\sc #1})} }
    \\ \bigfract{#2}{#3}
\end{array}}
\newcommand{\bigfract}[2]{\frac{^{\textstyle #1}}{_{\textstyle #2}}}
\def \mathax #1#2{\begin{array}{l} {\mbox{\scriptsize ({\sc #1})} } \\ #2
\end{array}}

\title{Amending Choreographies}

\author{Ivan Lanese 
\institute{Focus Team, University of Bologna/INRIA, Italy}
\email{lanese@cs.unibo.it} 
\and Fabrizio Montesi 
\institute{IT University of Copenhagen, Denmark}
\email{fmontesi@itu.dk}
\and Gianluigi Zavattaro
\institute{Focus Team, University of Bologna/INRIA, Italy}
\email{zavattar@cs.unibo.it} 
}

\def\titlerunning{Amending Choreographies}
\def\authorrunning{Lanese, Montesi, Zavattaro}

\begin{document}
\maketitle

\begin{abstract}
Choreographies are global descriptions of system behaviors, from
which the local behavior of each endpoint entity can be obtained
automatically through projection.  To guarantee that its projection is
correct, i.e.\ it has the same behaviors of the original choreography, a
choreography usually has to respect some coherency conditions. This
restricts the set of choreographies that can be projected.

In this paper, we present a transformation for amending choreographies
that do not respect common syntactic conditions for projection
correctness.  Specifically, our transformation automatically reduces
the amount of concurrency, and it infers and adds hidden
communications that make the resulting choreography respect the
desired conditions, while preserving its behavior.
\end{abstract}

\section{Introduction}
Choreography-based programming is a powerful paradigm where the programmer
defines the communication behavior of a system from a global viewpoint,
instead of separately specifying the behavior of each endpoint entity.
Then, the local behavior of each endpoint can be automatically generated
through a notion of 
\emph{projection}~\cite{hondaESOP,hondaPOPL,SEFM08,chorDATA,chor:popl13}.
A projection procedure is usually a homomorphism from choreographies
to endpoint code.  However, projections of some choreographies may
lead to undesirable behavior. To characterize the class of
choreographies that can be correctly projected some syntactic
conditions have been put forward.
Consider the following choreography $C$:
\[
C 
= 
\trans{o_1}{a}{b};\; \trans{o_2}{c}{d}
\]
Here, $\p a$, $\p b$, $\p c$ and $\p d$ are \emph{participants} and $o_1$ and
$o_2$ are \emph{operations} (labels for communications).
$C$ specifies a system where $\p a$ sends to $\p b$ a message on operation
$o_1$ ($\trans{o_1}{a}{b}$), and then (; is sequential composition)
$\p c$ sends to $\p d$ a message on operation
$o_2$ ($\trans{o_2}{c}{d}$).
We can naturally project $C$ onto a system $S$ composed of CCS-like
processes annotated with roles~\cite{SEFM08}:
\[
S 
= 
\role{\co{o_1}}{a} \parallel \role{o_1}{b} \parallel
\role{\co{o_2}}{c} \parallel
\role{o_2}{d}
\]
Here, $\role{\cdot}{a}$ specifies the behavior of participant $\p a$,
$\parallel$ is parallel composition, $\co{o_1}$ denotes an output on
operation $o_1$, and $o_1$ the corresponding input.  Unfortunately,
the projected system $S$ does not implement $C$ correctly.  Indeed, $\p c$
could send its message on $o_2$ to $\p d$ before the communication on
$o_1$ between $\p a$ and $\p b$ has occurred.
This is in contrast with the intuitive meaning of the sequential composition operator $;$ in
the choreography, which explicitly models sequentiality between the two
communications.
In~\cite{SEFM08}, 
we provide syntactic conditions on choreographies
for ensuring that projection will behave correctly. Similar conditions are used in
other works (e.g., \cite{hondaESOP,hondaPOPL}).
Such conditions would reject choreography $C$ above, by recognizing that it is
not \emph{connected}, i.e., the order of communications specified in $C$ cannot
be enforced by
the projected roles.

In this paper, we present a procedure for automatically transforming, or
\emph{amending}, a choreography that is not connected into a 
behaviorally equivalent one that is connected.  Our transformation
acts in two ways. In most of the cases, it automatically infers and
adds some extra hidden communications. In a few cases instead the
problem cannot be solved by adding communications, and the amount of
concurrency in the system has to be decreased.  For example, we can
transform $C$ in the connected choreography $C'$ below:
\[
C' 
= 
\trans{o_1}{a}{b};\; \trans{o_3^*}{b}{e};\;
\trans{o_4^*}{e}{c};\; \trans{o_2}{c}{d}
\]
Here, we have added an extra role $\p e$ that interacts with $\p b$ and $\p c$
to ensure the sequentiality of their respective communications.
This is the first kind of transformation discussed above.
We refer to Example~\ref{ex:par} in Section~\ref{sec:combining} for an example of the second kind.

Our transformation brings several benefits.
First, designers could use choreographies by defining only the desired
behavior, leaving to our transformation the burden of filling in the necessary
details on how the specified orderings should be enforced.
Moreover, our solution does not change the standard definition
of projection, since we operate only at the level of choreographies.
Hence, we retain simplicity in the definition of projection.
Finally, our transformation offers an automatic way of ensuring
correctness when composing different choreographies developed
separately into a single one.

\paragraph{Structure of the paper:} Section~\ref{sec:choreo} introduces
choreographies and their semantics. Section~\ref{sec:amending}
contains the main contribution of the paper, namely the description
and the proof of correctness of the amending
techniques. Section~\ref{sec:appl} applies the developed theory to a
well-known example, the two-buyers protocol. Finally,
Section~\ref{sec:concl} discusses related work and future directions.
Appendix~\ref{sec:orch} and Appendix~\ref{sec:proj} describe
projected systems and the projection operation, respectively.

\section{Choreography Semantics}\label{sec:choreo}
We introduce here our language for modeling choreographies.
The set of participants in a choreography, called \emph{roles}, is ranged over
by $\p a$, $\p b$, $\p c$, $\ldots$. We also consider two kinds of
\emph{operations} for communications: \emph{public operations},
ranged over by $o$, which represent observable activities of the
system, and \emph{private operations}, ranged over by $o^*$, used for
internal synchronization.  We use $o^?$ to range over both public and
private operations.

Choreographies, ranged over by $C$, $C'$, $\ldots$, are
defined as follows:
$$C ::= \trans{o^?}{a}{b} \mid \one 
\mid \zero 
\mid C; C' \mid C \parallel C' \mid C + C'$$
%
An interaction $\trans{o^?}{a}{b}$ means that role $\p a$ sends a
message on operation $o^?$ to role $\p b$ (we assume that $\p
a \neq \p b$).  Besides, there are the empty choreography $\one$, the
deadlocked choreography $\zero$, sequential and parallel composition
of choreographies, and nondeterministic choice between
choreographies. Deadlocked choreography $\zero$ is only used at
runtime, and it is not part of the user syntax.
%

\begin{figure}[t]
\[
\begin{array}{c}
\mathax{Interaction}{\trans{o^?}{a}{b} \arro{\transl{o^?}{a}{b}} \one}\qquad
\mathax{End}{\one \arro{\tick} \zero}\qquad
\mathrule{Sequence}{C \arro{\sigma} C' \quad \sigma \neq
\tick}{C;C'' \arro{\sigma} C';C'' }
\qquad
\mathrule{Parallel}{C \arro{\sigma} C' \quad \sigma \neq
\tick}{C\parallel C'' \arro{\sigma} C' \parallel C''} 
\\
\mathrule{Choice}{C \arro{\sigma} C'}{C+C'' \arro{\sigma} C'}
\qquad
\mathrule{Seq-end}{C_1 \arro{\tick} C_1' \quad C_2
\arro{\sigma} C_2'}{C_1;C_2 \arro{\sigma} C_2'} \qquad
\mathrule{Par-end}{C_1 \arro{\tick} C_1' \quad C_2
\arro{\tick} C_2'}{C_1\parallel C_2 \arro{\tick}
C_1'\parallel C_2'}
\end{array}
\]
\caption{Choreographies, semantics (symmetric rules
omitted).}\label{fig:ioclts} 
\end{figure}
The semantics of choreographies is the smallest labeled transition
system (LTS) closed under the rules in Figure~\ref{fig:ioclts}. Symmetric
rules for parallel composition and choice have been omitted. We use
$\sigma$ to range over labels. We have two kinds of labels: label
$\trans{o^?}{a}{b}$ denotes the execution of an interaction
$\trans{o^?}{a}{b}$ while label $\tick$ represents the termination of
the choreography. Rule {\sc Interaction} executes an interaction. Rule
{\sc End} terminates an empty choreography. Rule {\sc Sequence}
executes a step in the first component of a sequential
composition. Rule {\sc Parallel} executes an interaction from a
component of a parallel composition, while rule {\sc Choice} starts the
execution of an alternative in a nondeterministic choice. Rule {\sc
Seq-end} acknowledges the termination of the first component of a
sequential composition, starting the second component. Rule {\sc
Par-end} synchronizes the termination of two parallel components.

We use the semantics to build some notions of traces for
choreographies, which will be used later on for stating our results.

\begin{definition}[Choreography traces]
A (strong maximal) trace of a choreography $C_1$ is a sequence of
labels $\tilde\sigma = \sigma_1, \dots, \sigma_n$ such
that there is a sequence of
transitions $C_1 \arro{\sigma_1} \dots \arro{\sigma_n}
C_{n+1}$ and that $C_{n+1}$ has no outgoing
transitions.

A weak trace of a choreography $C$ is a sequence of labels
$\tilde\sigma$ obtained by removing all the labels of the form
$\trans{o^*}{a}{b}$ (corresponding to private interactions) from a
strong trace of $C$.

Two choreographies $C$ and $C'$ are (weak) trace equivalent iff they have the same set of (weak) traces.
\end{definition}

\section{Amending Choreographies}\label{sec:amending}
In this section we consider three different connectedness properties,
and we show how to enforce each one of them, preserving the set of
weak traces of the choreography. We refer to~\cite{TR} (a technical
report extending~\cite{SEFM08}) for the description of why these
properties guarantee projectability. Actually, in~\cite{SEFM08},
different notions of projectability are considered according to
whether the semantic model is synchronous or asynchronous, and on
which behavioral relation between a choreography and its projection is
required. The conditions presented and enforced below are correct and
ensure projectability according to all the notions considered
in~\cite{SEFM08}.

All the conditions below are stated at the level of
choreographies. Nevertheless, definitions of projection, projected
system and of its semantics may help in understanding why they are
needed. We collect these definitions in the Appendix.

We discuss the connectedness conditions in increasing order of
difficulty to help the understanding. Then, in
Section~\ref{sec:combining}, we combine them to make a general
choreography connected.

\subsection{Connecting Sequences}
\emph{Connectedness for sequence} ensures that, in a sequential composition
$C;C'$, the choreography $C'$ starts its execution only after $C$
terminates.  To formalize this property, we first need to define
auxiliary functions $\transI$ and $\transF$, which compute
respectively the sets of initial and final interactions in a
choreography:
\[
\begin{array}{l}
\transI(\trans{o^?}{a}{b}) = \transF(\trans{o^?}{a}{b}) = \{\trans{o^?}{a}{b}\}
\\
\transI(\one)=\transI(\zero)=\transF(\one)=\transF(\zero)=\emptyset
\\
\transI(C \parallel C') = \transI(C +
C') = \transI(C) \cup \transI(C')
\\
\transF(C \parallel C') = \transF(C +
C') = \transF(C) \cup \transF(C')
\\
\transI(C;C')  = \transI(C) \cup \transI(C') \text{ if } \exists C'' \textrm{ such that } C 
\arro{\tick} C'',\ \ \transI(C) \text{ otherwise}\\

\transF(C;C')  = \transF(C) \cup \transF(C')
\text{ if } \exists C'' \textrm{ such that } C'
\arro{\tick} C'',\ \ \transF(C') \text{ otherwise}
\end{array}
\]
Intuitively, connectedness for sequence can be
ensured by guaranteeing that a role waits for termination of all the
interactions in $C$ and only then starts the execution of
the interactions in $C'$.
We can now formally state this property.
\begin{definition}[Connectedness for sequence]
A choreography $C$ is \emph{connected for sequence} iff each
subterm of the form $C';C''$ satisfies $\forall
\trans{o_1^?}{a}{b} \in \transF(C')$, $\forall
\trans{o_2^?}{c}{d} \in \transI(C''). \p b=\p c$.
\end{definition}
%

Our transformation reconfigures the
subterms $C';C''$ failing to meet the condition. Take one such
term $C';C''$.  Choose a fresh role
$\p e$. Consider all the interactions $\trans{o^?}{a}{b}$ contributing to
$\transF(C')$ in the term. For each of them choose a fresh
operation $o_f^*$ and replace $\trans{o^?}{a}{b}$ with
$\trans{o^?}{a}{b};\;\trans{o_f^*}{b}{e}$. Similarly, for each
interaction $\trans{o^?}{c}{d}$ contributing to $\transI(C'')$ choose a fresh
operation $o_f^*$ and replace
$\trans{o^?}{c}{d}$ with $\trans{o_f^*}{e}{c};\;\trans{o^?}{c}{d}$.  The
essential idea is that role $\p e$ waits for all the threads in
$C'$ to terminate before starting threads in $C''$.

\begin{proposition}\label{prop:connectingseq}
Given a choreography $C$ we can derive using the pattern
above a choreography $C'$ which is connected for
sequence such that $C$ and $C'$ are weak trace
equivalent.
\end{proposition}
\begin{proof}
We have to apply the pattern to all the subterms failing to satisfy the
condition. We start from smaller subterms, and then move to larger
ones. We have to show that at the end all the subterms satisfy the
condition. The new subterms introduced by the transformation have the
form $\trans{o^?}{a}{b};\trans{o^*}{b}{e}$ and
$\trans{o^*}{e}{c};\trans{o^?}{c}{d}$, thus they satisfy the condition
by construction. Let us consider a subterm $D';D''$ obtained by
transforming a subterm $E';E''$ which already satisfied the condition (but whose subterms may have been transformed).
It is easy to check that $\transF(D') = \transF(E')$ and $\transI(D'')
= \transI(E'')$, thus the term is still connected for sequence.  For
subterms obtained by transforming subterms that did not satisfy the
condition, the thesis holds by construction.

Weak trace equivalence is ensured since only interactions on private
operations, which have no impact on weak trace equivalence, are added
by the transformation.
\end{proof}

\subsection{Connecting Choices}
\emph{Unique points of choice} ensure that in a choice all the participants
agree on the chosen branch. Intuitively, unique points of choice can
be guaranteed by ensuring that a single participant performs the
choice and informs all the other involved participants.

\begin{definition}[Unique points of choice]
A choreography $C$ has \emph{unique points of choice} iff each subterm of the form $C'+C''$ satisfies:
\begin{enumerate}
\item\label{ch:conn} $\forall \trans{o_1^?}{a}{b} \in \transI(C'),\forall \trans{o_2^?}{c}{d}
\in \transI(C''). \p a=\p c$; 
\item\label{ch:roles} $\roles(C')=\roles(C'')$; 
\end{enumerate}
where $\roles(\bullet)$ computes the set of roles in a choreography.
\end{definition}
We present below a pattern able to ensure unique points of choice.
We apply the pattern to all the subterms of the form $C'+C''$ that do not
satisfy one of the conditions. Take one such term $C'+C''$.  If
condition \ref{ch:conn} 
is not satisfied then choose
a fresh role $\p e$.  Consider all the interactions
$\trans{o^?}{a}{b}$ contributing to $\transI(C')$ or to
$\transI(C'')$. For each of them choose a fresh operation $o^*_f$ and
replace $\trans{o^?}{a}{b}$ with
$\trans{o_f^*}{e}{a};\;\trans{o^?}{a}{b}$.

Suppose now that condition \ref{ch:conn} 
is satisfied, while
condition \ref{ch:roles} 
is not. Then we
can assume a role $\p e$ which is the sender of all the interactions in
$\transI(C'+C'')$. Consider each role $\p a$ that occurs in $C'$ but
not in $C''$ (the symmetric case is analogous). For each of them add in
parallel to $C''$ the interaction $\trans{o_f^*}{e}{a}$ where $o_f^*$
is a fresh operation.

\begin{proposition}\label{prop:connectingchoice}
Given a choreography $C$ we can derive using the pattern
above a choreography $C'$ which has unique
points of choice 
such that $C$ and $C'$ are weak
trace equivalent.
\end{proposition}
\begin{proof}
We have to apply the pattern to all the subterms failing to satisfy the
condition. We start from smaller subterms, and then move to larger
ones. We have to show that at the end all the subterms satisfy the
conditions.  We consider the transformation ensuring
condition \ref{ch:conn} first,
then the one ensuring condition \ref{ch:roles}.

For the first transformation, given a subterm $D'+D''$ there are two
cases: either some initial interactions inside $D'$ and $D''$ have been
modified or not. In the second case the thesis follows by inductive
hypothesis. In the first case all the initial
interactions have been changed, and the freshly introduced role is the
new sender in all of them. 
Thus the condition is satisfied.

Let us consider the second transformation. The transformation has no
impact on subterms, since it only adds interactions in parallel. For
the whole term the thesis holds by construction. Note that the term
continues to satisfy the other condition.

Weak trace equivalence is ensured since only interactions on private
operations, which have no impact on weak trace equivalence, are added
by the transformation.
\end{proof}

\subsection{Connecting Repeated Operations}
If different interactions use the same operation, one has to ensure
that messages do not get mixed, i.e.\ that the output of one
interaction is not matched with the input of a different interaction
on the same operation. Since in our model outputs do not contain the
target participant, and inputs do not contain the expected sender, the
problem is particularly relevant. The same problem however may occur
also if the output contains the target participant and the input the
expected sender, but only for interactions between the same
participants, as shown by the example below.

\begin{example}\label{ex:sameroles}
Consider the choreography $C$ below:
\[
C = (\trans{o_1}{a}{b};\trans{o}{b}{c};\trans{o_2}{c}{d})
\parallel 
(\trans{o_3}{a}{b};\trans{o}{b}{c};\trans{o_4}{c}{d})
\]
The two interactions $\trans{o}{b}{c}$, both between role $\p b$ and role $\p c$,
may interfere, since the send from one of them could be received by
the other one.  
\end{example}

Clearly, given two interactions on the same operation, the
problem does not occur if the output of the first is never enabled
together with the input of the second, and vice versa.

To ensure this, in \cite{SEFM08} we required causal dependencies
between the input in one interaction and the outputs of different
interactions on the same operation, and vice versa.  To formalize this
concept we introduce the notion of \emph{event}.  For each interaction
$\trans{o^?}{a}{b}$ we distinguish two events: a sending event at role
$\p a$ and a receiving event at role $\p b$. The sending event and the
receiving event in the same interaction are called matching events. We
use $e$ to range over events, $s$ to range over sending events and $r$
to range over receiving events. We denote by $\co{e}$ the event
matching event $e$. If event $e$ precedes event $e'$ in the causal
relation then $e'$ can become enabled only after $e$ has been
executed, thus they cannot be matched.  However, this requirement
alone is not suitable for our aims, since this condition is too strong
to be enforced by a transformation preserving weak traces.

Luckily, also events in a suitable relation of conflict cannot be
matched.  Requiring that events are in opposite branches of a choice,
however, is not enough, since they may be both enabled. For instance,
in $\trans{o}{a}{b}+\trans{o}{c}{d}$ the output from $\p a$ can be
captured by $\p d$\footnote{This choreography however has not unique
  points of choice.}.
However,
if the role containing the event is already aware of the choice, then
the events are never enabled together and thus cannot be matched,
since when one of them becomes enabled the other one has already been
discarded. We call \emph{full conflict} such a relation. For instance,
in 
$$(\trans{o'}{a}{b};\trans{o}{b}{a};
\trans{o'}{a}{c})+(\trans{o''}{a}{b};\trans{o}{b}{a};\trans{o''}{a}{c})$$
the two interactions on $o$ do not interfere. Note that an
interference could cause the wrong continuation to be chosen.

Thus we require here that a send and a receive in different
interactions but on the same operation are either in a causality
relation or in a full conflict relation. We call this condition
\emph{causality safety}. Causality safety can be enforced by a
transformation preserving weak traces. We refer to~\cite{TR} for the
proof that causality safety ensures the desired properties of
choreography projection.

To define causality safety we thus need to define a \emph{causality relation}
and a \emph{full conflict relation}.

\begin{definition}[Causality relation]
Let us consider a choreography $C$.  A \emph{causality relation} $\leq_C$ is
a partial order among events of $C$.  We define $\leq_C$ as the
minimum partial order satisfying:
\begin{description}
\item[sequentiality:] for each subterm of the form $C';C''$ and each
  role $\p a$, if $r'$ is a receive event in $C'$ at role $\p a$, $e''$
  is a generic event in $C''$ at role $\p a$
  then $r' \leq_C e''$;
\item[synchronization:] for each receive event $r$ and generic event
  $e'$, $r \leq_C e'$ implies $\co{r} \leq_C e'$.
\end{description}
\end{definition}  

\begin{definition}[Full conflict relation]\label{defin:synchconf}
Let us consider a choreography $C$.  
A \emph{full conflict relation} $\fconf_C$ is a relation among events of $C$.
We define $\fconf_C$ as the smallest symmetric relation satisfying:
\begin{description}
\item[choice:] for each subterm of the form $C'+C''$ and each
  role $\p a$, if $e'$ is an event in $C'$ at role $\p a$, $e''$
  is an event in $C''$ at role $\p a$ then $e' \fconf_C e''$;
\item[causality:] if $e' \fconf_C e''$ and $e'' \leq_C e'''$ then $e'
  \fconf_C e'''$.
\end{description}
\end{definition}

We can now define causality safety.

\begin{definition}[Causality safety]
A choreography $C$ is \emph{causality safe} iff for each pair of interactions
$\trans{o^?}{a}{b}$ (with events $s_1$ and $r_1$) and
$\trans{o^?}{c}{d}$ (with events $s_2$ and $r_2$) on the same
operation $o^?$ we have that:
\begin{itemize}
\item $s_1 \leq_C r_2 \lor r_2 \leq_C s_1 \lor s_1 \fconf_C r_2$; 
\item $s_2 \leq_C r_1 \lor r_1 \leq_C s_2 \lor s_2 \fconf_C r_1$.
\end{itemize}
\end{definition}

As an example, choreography $C = \trans{o}{a}{b} \parallel
\trans{o}{c}{d}$ is not causality safe, since the output on $o$ at $\p
a$ is enabled together with the input on $o$ at $\p d$, thus enabling
a communication from $\p a$ to $\p d$ which should not be allowed.

A causality safety issue is immediately solved by renaming one of the
operations.  However, this changes the specification. We show how
to solve the causality safety issue while sticking to the original
(weak) behavior. We have different approaches according to the
top-level operator of the smaller term including the conflicting
interactions.  Thus we distinguish \emph{parallel causality safety},
\emph{sequential causality safety} and \emph{choice causality safety}.

To solve parallel causality safety issues we apply a form of
\emph{expansion law} that transforms the parallel composition into
nondeterminism, thus either removing completely the causality safety
issue or transforming it into sequential or choice causality safety,
discussed later on.  Using the expansion law one can transform any
choreography into a \emph{normal form} defined as below.

\begin{definition}[Normal form]
A choreography $C$ is in \emph{normal form} if it is written as:
$$\sum_{i} \trans{o^?_i}{a_i}{b_i};{C}_i$$ 
where $\sum_i$ is $n$-ary nondeterministic choice (we can see the
empty sum as $\zero$), and $C_i$ is in normal form for each $i$. 
\end{definition}

The expansion law is defined below.

\begin{definition}[Expansion law]
\begin{eqnarray*}
(\sum_{i} \trans{o^?_i}{a_i}{b_i};\;C_i)\parallel(\sum_{j}
\trans{o^?_j}{a_j}{b_j};C_j) & = & (\sum_{i}
\trans{o^?_i}{a_i}{b_i};(C_i\parallel(\sum_{j}
\trans{o^?_j}{a_j}{b_j};C_j)))\\ & & + (\sum_{j}
\trans{o^?_j}{a_j}{b_j};(C_j \parallel (\sum_{i}
\trans{o^?_i}{a_i}{b_i};C_i)))
\end{eqnarray*}
\end{definition}

The expansion law is correct w.r.t.\ trace equivalence, in the
sense that applying the expansion law does not change the
set of traces (neither strong nor weak).

\begin{lemma}\label{lemma:expansion}
Let $C$ and $C'$ be choreographies, with $C'$ obtained by
applying the expansion law to a subterm of $C$. Then
$C$ and $C'$ are (strong and weak) trace equivalent.
\end{lemma}
\begin{proof}
Labels not involving the subterm are easily mimicked. Consider the
first label involving the subterm. If no such label exists then the
thesis follows. Otherwise, the label corresponds to the execution of
one of the interactions $\trans{o^?_i}{a_i}{b_i}$ or
$\trans{o^?_j}{a_j}{b_j}$. Executing any of these interactions reduces
both the terms to the same term. The thesis follows.
\end{proof}

Using the expansion law we can put any choreography $C$ in
normal form.

\begin{proposition}[Normalization]\label{prop:normform}
Given a choreography $C$ there is a choreography $C'$ in normal
form such that $C$ and $C'$ are weak trace equivalent.
\end{proposition}
\begin{proof}
The proof is by structural induction on the number of interactions
occurring in $C$. The cases of interactions and $\zero$ are
trivial. Choreography $\one$ can be replaced by any
interaction on a private operation without changing the set of weak traces.  For sequential
composition note that $(\sum_{i} \trans{o^?_i}{a_i}{b_i};C_i);C'$ and
$(\sum_{i} \trans{o^?_i}{a_i}{b_i};C_i;C')$ have the same set of
traces. $C_i;C'$ can be transformed in normal form by inductive
hypothesis. For nondeterministic choice the thesis is trivial (it is
easy to check that nondeterministic choice is associative). For
parallel composition one can apply the expansion law, and the thesis
follows from Lemma~\ref{lemma:expansion} and inductive hypothesis.
\end{proof}

Let us now consider sequential causality safety. The term should have
the form $C';C''$, with an interaction $\trans{o^?}{a}{b}$ in $C'$ and
an interaction $\trans{o^?}{c}{d}$ in $C''$. If the term satisfies
connectedness for sequence then the send at $\p c$ cannot become enabled
before the receive at $\p b$. We thus have to ensure that the receive at
$\p d$ becomes enabled after the message from the send at $\p a$ has been
received. Choose fresh operations $o_f^*$ and $o_g^*$ and replace
$\trans{o^?}{a}{b}$ with
$\trans{o^?}{a}{b};\trans{o_f^*}{b}{d};\trans{o_g^*}{d}{b}$.

The pattern for choice causality safety is similar. The term should
have the form $C'+C''$, with an interaction $\trans{o^?}{a}{b}$ in
$C'$ and an interaction $\trans{o^?}{c}{d}$ in $C''$. Assume that
there is an interference between the send at $\p a$ and the receive at
$\p d$ (the symmetric case is similar). 
Choose fresh operations $o_f^*$ and $o_g^*$ and replace
$\trans{o^?}{c}{d}$ with
$\trans{o_f^*}{c}{d};\trans{o_g^*}{d}{c};\trans{o^?}{c}{d}$


To prove the correctness of the transformation we need two auxiliary lemmas.

\begin{lemma}\label{lemma:seqcaussend}
Let $C$ be a choreography which is connected for sequence. Then all sending events in $C$ depend on sending events in initial interactions in $C$.
\end{lemma}
\begin{proof}
By structural induction on $C$. The only difficult case is sequential
composition, when $C=C';C''$. For events in $C'$ the thesis follows by
inductive hypothesis. For events in $C''$, if they are in an initial
interaction, from connectedness for sequence they are in the same role
as receiving events in the final interactions in $C'$, thus they
depend on them. By synchronization they also depend on sending events
in $C'$. The thesis follows by induction and by transitivity. For
events not in initial interactions in $C''$, by inductive hypothesis
they depend on events in initial interactions in $C''$, thus the thesis follows
from what proved above by transitivity.
\end{proof}


\begin{lemma}\label{lemma:issueseq}
Let $C=C';C''$ be a choreography which is connected for sequence. Then
there are no causality safety issues between receiving events in $C'$
and sending events in $C''$.
\end{lemma}
\begin{proof}
From connectedness for sequence all the receiving events in final
interactions of $C'$ and all the sending events in initial
interactions of $C''$ are performed by the same role. Thus, from the
sequentiality condition in the definition of causality relation, they
are causally related. Now the more difficult case is when the send $s$
in $C''$ is not initial and the receive $r$ in $C'$ is not
final. Thanks to Lemma~\ref{lemma:seqcaussend} the send $s$ depends on
an initial send, and by transitivity on each final receive. Thanks to
synchronization it depends also on each final send.  Assume $r$ is not
final. Thanks to connectedness for sequence it is in the same role of
a following send $s'$, thus $s'$ depends on $r$. If $s'$ is final the
thesis follows. Otherwise we iterate the procedure on a smaller term,
and the thesis follows by induction.
\end{proof}

We can now prove the correctness of the transformation.

\begin{proposition}\label{prop:connectingrep}
Let $C$ be a choreography which is connected for sequence, has unique
points of choice and is parallel causality safe. 
We can derive using
the pattern above a choreography $C'$ which has no sequential or
choice causality safety issue such that $C$ and $C'$ are weak trace
equivalent. Also, $C'$ is still connected for sequence, has unique
points of choice and is parallel causality safe.
\end{proposition}
\begin{proof}
We prove the thesis for just one causality safety issue, the more
general result follows by iteration. Take a causality safety issue,
and consider the smallest subterm including the two conflicting
interactions. We have two cases according to whether the subterm has
the form $C';C''$ or $C'+C''$.

Let us consider the first case. Thanks to Lemma~\ref{lemma:issueseq} the issue
is between the send at $\p a$ in an interaction $\trans{o^?}{a}{b}$ in $C'$ and a
receive at $\p d$ in an interaction $\trans{o^?}{c}{d}$ in $C''$. 

After the transformation, the receive at $\p d$ in $\trans{o^?}{c}{d}$
depends on the receive at $\p d$ in $\trans{o_f^*}{b}{d}$ by
sequentiality. The receive at $\p d$ in $\trans{o_f^*}{b}{d}$ depends on
the send at $\p b$ inside the same interaction by synchronization, on the
receive at $\p b$ in $\trans{o^?}{a}{b}$ by sequentiality, and on the
send at $\p a$ inside the same interaction again by synchronization.
This proves that the issue is solved.
%
%
%

Let us consider the second case. Assume that the issue is between a
send in an interaction $\trans{o^?}{a}{b}$ in $C'$ and a receive in an
interaction $\trans{o^?}{c}{d}$ in $C''$. The sends at $\p a$ and at $\p c$
are causally dependent on sends in an initial interaction of the term
thanks to Lemma~\ref{lemma:seqcaussend}. The transformation adds a
dependency between the send at $\p c$ and the receive at $\p d$ inside the
same interaction. Thus the send at $\p a$ and the receive at $\p d$ are
in full conflict, since they are both causally dependent on sends from
initial interactions of the term, which are all in the same role
thanks to unique points of choice.
Again, this proves that the issue is
solved.

Weak trace equivalence is ensured since only interactions on private
operations, which have no impact on weak trace equivalence, are added
by the transformation.

We have to show that the other properties are preserved. All of them
hold by construction for the newly introduced terms.  Connectedness
for sequence is preserved since the transformation preserves senders
of initial interactions and receivers of final interactions. Unique
points of choice are preserved for the same reason, and since the
transformation does not add new roles. Parallel causality safety is
preserved since the transformation only introduces interactions on
fresh operations.
\end{proof}


\subsection{Combining the Amending Techniques}\label{sec:combining}
Till now we have shown that given a choreography which fails to
satisfy one of the connectedness conditions, we can transform it into
an equivalent one that satisfies the chosen connectedness condition.
Some care is required to avoid that, while ensuring that one condition
is satisfied, violations of other conditions are introduced. In fact,
if such a situation would occur repeatedly, the connecting procedure
may not terminate.

The following theorem proves that we can combine the connecting
patterns presented in the previous sections to define a terminating
algorithm transforming any choreography into a connected choreography.

\begin{theorem}[Connecting choreographies]
There is a terminating procedure that given any choreography
$C$ creates a new choreography $D$ such that:
\begin{itemize}
\item $D$ is connected;
\item $C$ and $D$ are weak trace equivalent.
\end{itemize}
\end{theorem}
\begin{proof}
We can apply the normalization procedure to all the subterms of
choreography $C$ that do not satisfy parallel causality safety,
starting from the smallest subterms to the largest, to get a
choreography $C'$ which is parallel causality safe (since the
undesired parallel compositions have been removed) and which is weak
trace equivalent to $C$ thanks to Proposition~\ref{prop:normform}.

Now, again from the smallest subterms to the largest, we can apply to
$C'$ the procedure for ensuring unique points of choice 
to those subterms which have a top-level
nondeterministic choice operator, and the procedure for making them
connected for sequence to those subterms which have a
top-level sequential composition operator obtaining a choreography
$C''$.

For terms of the first kind, thanks to
Proposition~\ref{prop:connectingchoice}, we obtain terms which have
unique points of choice. 
They are also
parallel causality safe and connected for sequence, since the
transformation may not create these issues.
The same holds for terms of the second kind thanks to
Proposition~\ref{prop:connectingseq}.  In both the cases, the
resulting term is weak trace equivalent to the starting one.
By transitivity $C''$ is weak trace equivalent to $C$.

Finally, to each sequential and choice causality safety issue we can
apply the procedure to solve it. From Proposition~\ref{prop:connectingrep}
we know that the resulting choreography $D$ has no sequential or
choice causality safety issue, still satisfies the other properties
and is weak trace equivalent to $C''$. The thesis follows by transitivity.
\end{proof}

\begin{example}\label{ex:par}
We now apply our procedure to the paradigmatic choreography $C=
\trans{o}{a}{b} \parallel \trans{o}{c}{d}$. First note that $C$ does
not satisfy parallel causality safety.  By application of the
expansion law we obtain: 
$$C_1= \trans{o}{a}{b};\trans{o}{c}{d} +
\trans{o}{c}{d};\trans{o}{a}{b}$$

Proceeding from smaller to larger subterms, we first encounter the
subterms $\trans{o}{a}{b};\trans{o}{c}{d}$ and
$\trans{o}{c}{d};\trans{o}{a}{b}$ which are not connected for sequence
(and are not sequential causality safe).  By applying the
corresponding pattern to the two subterms, we obtain: 
$$C_2=
\trans{o}{a}{b};
\trans{o_1^*}{b}{e'};\trans{o_2^*}{e'}{c};\trans{o}{c}{d} +
\trans{o}{c}{d}; \trans{o_3^*}{d}{e''};\trans{o_4^*}{e''}{a};
\trans{o}{a}{b}$$  

Now, the whole term does not have
unique points of choice,
and is not choice causality safe.  By applying the first of the transformations
ensuring unique points of choice, we obtain:
$$
\begin{array}{lll}
C_3 & = & \trans{o_5^*}{e}{a};\trans{o}{a}{b};
\trans{o_1^*}{b}{e'};\trans{o_2^*}{e'}{c};\trans{o}{c}{d} +\\
&& \trans{o_6^*}{e}{c};\trans{o}{c}{d};
\trans{o_3^*}{d}{e''};\trans{o_4^*}{e''}{a};
\trans{o}{a}{b}
\end{array}
$$
By applying the transformation ensuring 
that both the branches have the same set of roles,
we obtain:
$$
\begin{array}{lll}
C_4 & = & \big( \trans{o_5^*}{e}{a};\trans{o}{a}{b};
\trans{o_1^*}{b}{e'};\trans{o_2^*}{e'}{c};\trans{o}{c}{d} 
\parallel \trans{o_7^*}{e}{e''} \big)\ 
+ \\
& & \big( \trans{o_6^*}{e}{c};\trans{o}{c}{d};
\trans{o_3^*}{d}{e''};\trans{o_4^*}{e''}{a};
\trans{o}{a}{b} \parallel \trans{o_8^*}{e}{e'} \big)
\end{array}$$
We are now left with sequential and choice causality safety issues.
We start by solving the sequential ones, which are one for each
branch.  As expected, they are between the sender of the first
interaction on $o$ and the receiver of the second.
By applying the pattern we get:
$$
\begin{array}{lll}
C_5 & = & \big( \trans{o_5^*}{e}{a};\trans{o}{a}{b};\trans{o_9^*}{b}{d};\trans{o_{10}^*}{d}{b};
\trans{o_1^*}{b}{e'};\trans{o_2^*}{e'}{c};\trans{o}{c}{d} 
\parallel \trans{o_7^*}{e}{e''} \big)\ 
+ \\
& & \big( \trans{o_6^*}{e}{c};\trans{o}{c}{d};\trans{o_{11}^*}{d}{b};\trans{o_{12}^*}{b}{d};
\trans{o_3^*}{d}{e''};\trans{o_4^*}{e''}{a};
\trans{o}{a}{b} \parallel \trans{o_8^*}{e}{e'} \big)
\end{array}$$
We are left with choice causality safety issues between the interaction $\trans{o}{a}{b}$ in the first branch and the interaction $\trans{o}{c}{d}$ in the second branch. By applying the pattern we get:
$$
\begin{array}{lll}
C_6 & = & \big( \trans{o_5^*}{e}{a};\trans{o_{13}^*}{a}{b};\trans{o_{14}^*}{b}{a};\trans{o}{a}{b};\trans{o_9^*}{b}{d};\trans{o_{10}^*}{d}{b};\\
&&\quad\trans{o_1^*}{b}{e'};\trans{o_2^*}{e'}{c};\trans{o}{c}{d} 
\parallel \trans{o_7^*}{e}{e''} \big)\ 
+ \\
& & \big( \trans{o_6^*}{e}{c};\trans{o_{15}^*}{c}{d};\trans{o_{16}^*}{d}{c};\trans{o}{c}{d};\trans{o_{11}^*}{d}{b};\trans{o_{12}^*}{b}{d};\\
&&\quad\trans{o_3^*}{d}{e''};\trans{o_4^*}{e''}{a};
\trans{o}{a}{b} \parallel \trans{o_8^*}{e}{e'} \big)
\end{array}$$
This example shows that solving a parallel causality safety issue may
require a considerable amount of auxiliary communications. Thus, it is
advised to change the name of the operation if possible. If not
possible, our approach will solve the problem. Other issues are solved
more easily, since they involve no duplication of terms.
\end{example}

\section{Application: Two-Buyers Protocol}\label{sec:appl}
We show now how our transformation for connecting choreographies can
be used as an effective design tool for programming multiparty
choreographies.  We model the example reported in~\cite{hondaPOPL},
the two-buyers protocol, where two buyers -- $\p b_1$ and $\p b_2$ --
combine their finances for buying a product from a seller $\p s$.  The
protocol starts with $\p b_1$ asking the price for the product of
interest to $\p s$. Then, $\p s$ communicates the price to both $\p b_1$ and
$\p b_2$.  Subsequently, $\p b_1$ notifies $\p b_2$ of how much she is willing
to contribute to the purchase.  Finally, $\p b_2$ chooses whether to confirm the 
purchase and ask $\p s$ to send a 
delivery date for the product; otherwise, the choreography terminates.
%
We do not deal here with how this
choice is performed, as our choreographies abstract from data.

To create a quick prototype choreography $C$ for the two-buyers
protocol, we focus only on the main interactions. When writing the
choreography we do not worry about connectedness conditions, since we
will enforce them automatically later on. The code follows naturally:
$$
\begin{array}{ll}
C=&\trans{\op{price}}{b_1}{s};\ 
(\ \trans{\op{quote1}}{s}{b_1} \ \parallel \  \trans{\op{quote2}}{s}{b_2} \ );\ 
\trans{\op{contrib}}{b_1}{b_2}; \\ 
&(\ \trans{\op{ok}}{b_2}{s}; \trans{\op{delivery}}{s}{b_2} \ + \ \one \ )
\end{array}
$$ 
The code above is just a direct translation of our explanation in
natural language into a choreography.  We can immediately observe that
the choreography is not connected, since it contains $2$ violations of the connectedness conditions:
\begin{itemize}
\item the subterm $(\ \trans{\op{quote1}}{s}{b_1} \ \parallel \ 
\trans{\op{quote2}}{s}{b_2} \ );\ \trans{\op{contrib}}{b_1}{b_2}$ is not 
connected for sequence; thus,
e.g., $\p b_1$ may send the $\op{contrib}$ message before $\p b_2$ receives the message for
$\op{quote2}$\footnote{This may happen only with the asynchronous semantics.};
\item the subterm $(\ \trans{\op{ok}}{b_2}{s}; \trans{\op{delivery}}{s}{b_2} \ + 
\ \one \ )$ has not unique
points of choice.
\end{itemize}
We can apply our transformation to amend our choreography prototype,
transforming it into a connected choreography which is weak trace
equivalent to $C$, obtaining:
$$
\begin{array}{ll}
C_1=&\trans{\op{price}}{b_1}{s};\ 
(\ \trans{\op{quote1}}{s}{b_1};\ \trans{o_1^{*}}{b_1}{e_1} \ \parallel \ 
\trans{\op{quote2}}{s}{b_2};\ \trans{o_2^{*}}{b_2}{e_1} \ );\\
&\trans{o_3^{*}}{e_1}{b_1};\  \trans{\op{contrib}}{b_1}{b_2}; \ 
(\ \trans{\op{ok}}{b_2}{s};
\trans{\op{delivery}}{s}{b_2} \ + (\ \one \parallel \trans{o_4^{*}}{b_2}{s}\ ) \ )
\end{array}
$$ 
The choreography $C_1$ above is connected, thus it can be projected, and
the projection will be trace equivalent to $C_1$ itself (thanks to
the results in~\cite{SEFM08}), and weak trace equivalent to the original
choreography $C$.

\section{Conclusions}\label{sec:concl}
In previous work~\cite{SEFM08} we have defined syntactic conditions
that guarantee choreography connectedness, i.e.\ that the endpoints
obtained by the natural projection of a choreography correctly
implement the global specification given by the choreography itself.
In this paper we have presented a procedure for amending non-connected
choreographies by reducing parallelism and adding interactions on
private operations in such a way that all the above conditions are
satisfied, while preserving the observational semantics.  To the best
of our knowledge, only two other papers consider the possibility to
add messages to a choreography in order to make it correctly
projectable \cite{WWW,tsc}.

In \cite{WWW} non-connected sequences are connected by adding a synchronization 
from every role involved in a final interaction before the 
sequential composition to every role involved in an initial
interaction after the sequential composition, while non-connected choices 
are modified by selecting a dominant role responsible for taking the
choice and communicating it to the other participants.
Our approach reduces the number of added synchronizations in the case of
sequential compositions, and allows for a symmetric projection that  
treats all the roles in the same way in the case of choices.
In \cite{WWW} there is no discussion about repeated operations, and
to the best of our understanding
of the paper this is problematic.
In fact, the choreography $C$ in Example~\ref{ex:sameroles}
(written according to our syntax, which is slightly 
different w.r.t.\ the one adopted in \cite{WWW})
is well-formed according to the conditions in \cite{WWW},
but it is not projectable.

{\em Collaboration} and {\em Sequence diagrams} are popular visual representations
of choreographies~\cite{UML}. In this area, the work more similar to ours is \cite{tsc},
where the problem of amending collaboration diagrams is addressed.
In collaboration diagrams 
the interactions are labels of edges in a graph which
has one node for each role, and one
edge for each pair of interacting roles.
Interactions are organized in threads representing
sub-choreographies. Messages are totally ordered
within the same thread, while a partial order 
can be defined among interactions belonging to
different threads. The problem of amending non-connected
choreographies is simplified in this context:
collaboration diagrams do not have a global choice composition
operator among sub-choreographies, and two distinct threads use
two disjoint sets of operations.
For this reason, only the problem of non-connected
sequences is addressed in \cite{tsc}.

Another paper amending choreographies is \cite{amendingDATA}, however
they consider a different problem. In fact, they consider
choreographies including annotations on the communicated data, and
they concentrate on amending those annotations.

Possible extensions of the present work include its application to
existing choreography languages, such as~\cite{chor:popl13}, where
however the intended semantics of sequential composition is weaker,
since independent interactions can be freely swapped. The main
difficulty in extending the approach is dealing with infinite
behaviors. Having an iteration construct should not be a problem:
roughly one has to make connected the choice between iterating and
exiting the iteration using the techniques for choice, and the sequential
execution of different iterations using the techniques for
sequence. General recursion would be more difficult to deal with,
since subterms belonging to different iterations may execute in
parallel. Adding data to our language would not change the issues
discussed in the present paper.

\bibliographystyle{eptcs}
\bibliography{wwv2013}

\appendix

\section{Projected Systems}\label{sec:orch}
We describe here the syntax and the operational semantics of \emph{projected
systems}. Projected systems, ranged over by $S$, $S'$, $\ldots$, are
composed by \emph{processes}, ranged over by $P$, $P'$, $\ldots$,
describing the behavior of participants,
\begin{eqnarray*}
P & ::= & o^? \mid \co{o^?} \mid \one \mid P;P'\mid \ P\mid P'\ \mid P+P' \mid \msg{o^?} \mid \zero \\
S & ::= & \role{P}{a}\mid S \parallel S'
\end{eqnarray*}

Processes include input action $o^?$ and output action $\co{o^?}$ on a
specific operation $o^?$ (either public or private), the empty process
$\one$, sequential and parallel composition, and nondeterministic
choice. The runtime syntax includes also messages $\msg{o^?}$, used in
the definition of the asynchronous semantics, and the deadlocked
process $\zero$. 
Projected systems are parallel compositions of roles. Each role has a role
name and executes a process. We require role names to be unique.

We define two LTS semantics for projected systems, one \emph{synchronous}
and one \emph{asynchronous}.  In the synchronous semantics input actions
and output actions interact directly, while in the asynchronous one
the sending event creates a message that, later, may interact with the
corresponding input, generating a receiving event.

The asynchronous LTS for projected systems is the smallest LTS closed
under the rules in Table~\ref{table:pocasynchlts}. We use $\gamma$ to
range over labels.  Symmetric rules for parallel compositions (both
internal and external) and choice have been omitted.

\begin{table}[t]
\[
\begin{array}{c}
\mathax{In}{o^? \arro{o^?}\one}\qquad
\mathax{Out}{\co{o^?} \arro{\co{o^?}} \one}\qquad
\mathax{Async-Out}{\msg{o^?} \arro{\msg{o^?}} \one}\qquad
\mathax{One}{\one \arro{\tick} \zero}
\\[3mm]
\mathrule{Sequence}{P \arro{\gamma} P' \quad \gamma \neq \tick}{P;Q \arro{\gamma} P';Q}
\qquad
\mathrule{Inner Parallel}{P \arro{\gamma} P' \quad \gamma \neq \tick}{P \mid Q \arro{\gamma} P' \mid Q}
\qquad
\mathrule{Choice}{P \arro{\gamma} P'}{P+Q \arro{\gamma} P'} \qquad
\mathrule{Seq-end}{P \arro{\tick} P' \quad Q \arro{\gamma} Q'}{P;Q \arro{\gamma} Q'}\\[3mm]
\mathrule{Inner Par-end}{P \arro{\tick} P' \quad Q \arro{\tick} Q'}{P \mid Q \arro{\tick} P' \mid Q'}
\qquad
\mathrule{Lift}{P \arro{\gamma} P' \quad \gamma \neq \co{o^?},\tick}{\role{P}{a} \arro{\gamma:\p a} \role{P'}{a}}\qquad
\mathrule{Lift-Tick}{P \arro{\tick} P'}{\role{P}{a} \arro{\tick} \role{P'}{a}}\qquad
\mathrule{Msg}{P \arro{\co{o^?}} P'}{\role{P}{a} \arro{\co{o^?}:\p a} \role{P' \mid \msg{o^?}}{a}}\\[3mm]
\mathrule{Synch}{S\arro{\msg{o^?}:\p a}S' \quad S'' \arro{o^?:\p b}S'''}{S \parallel S'' \arro{\trans{o^?}{a}{b}} S' \parallel S'''}
\qquad
\mathrule{Ext-Parallel}{S \arro{\gamma} S' \quad \gamma \neq \tick}{S \parallel S'' \arro{\gamma} S' \parallel S''}
\mathrule{Ext-Par-End}{S \arro{\tick} S' \quad S'' \arro{\tick} S'''}{S\parallel S'' \arro{\tick} S' \parallel S'''}
\end{array}
\]
\caption{Projected systems asynchronous semantics (symmetric rules omitted).}\label{table:pocasynchlts}
\end{table}

Rules {\sc In} and {\sc Out} execute input actions and output actions,
respectively. Rule {\sc Asynch-Out} makes messages available for a
corresponding input action. Rule {\sc One} terminates an empty
process. Rule {\sc Sequence} executes a step in the first component of a
sequential composition. Rule {\sc Inner Parallel} executes an action
from a component of a parallel composition, while rule {\sc Choice}
starts the execution of an alternative in a nondeterministic
choice. Rule {\sc Seq-end} acknowledges the termination of the first
component of a sequential composition, starting the second
component. Rule {\sc Inner Par-end} synchronizes the termination of
two parallel components.
Rule {\sc Lift} lifts actions to the system level, tagging them with
the name of the role executing them. Action $\tick$ instead is dealt
with by rule {\sc Lift-Tick}, which lifts it without adding the role
name. Outputs instead are stored as messages by rule {\sc
  Msg}. Rule {\sc Synch} synchronizes a message with the corresponding
input action, producing an interaction. Rule {\sc Ext-Parallel} allows
parallel systems to stay idle. Finally, rule {\sc Ext-Par-End}
synchronizes the termination of parallel systems.

The synchronous LTS for projected systems is the smallest LTS closed
under the rules in Table~\ref{table:pocasynchlts}, where rules {\sc
  Out}, {\sc Async-Out} and {\sc Msg} are deleted and the new rule
{\sc Sync-Out} below is added:
$$\mathax{Sync-Out}{\co{o^?} \arro{\msg{o^?}}_s \one}$$
This rule allows outputs in the synchronous semantics to send messages
that can directly interact with the corresponding input at the system
level.

Synchronous transitions are denoted as $\arro{\gamma}_s$ instead of
$\arro{\gamma}$, to distinguish them from the asynchronous ones. 

As for choreographies, we define \emph{traces}. We have different
possibilities: in addition to the distinction between strong and weak
traces, we distinguish \emph{synchronous} and \emph{asynchronous}
traces.

\begin{definition}[Projected systems traces]
A (strong maximal) synchronous trace of a projected system $S_1$ is a
sequence of labels $\gamma_1, \dots, \gamma_n$, where $\gamma_i$ is of
the form $\trans{o^?}{a}{b}$, or $\tick$ for each $i \in
\{1,\dots,n\}$, such that there is a sequence of synchronous 
transitions $S_1 \arro{\gamma_1}_s \dots \arro{\gamma_n}_s
S_{n+1}$ and such that $S_{n+1}$ has no outgoing
transitions of the same form.

A (strong maximal) asynchronous trace of a projected system $S_1$ is a
sequence of labels $\gamma_1, \dots, \gamma_n$, where $\gamma_i$ is of
the form $\co{o^?}:\p a$, $\trans{o^?}{a}{b}$, or $\tick$ for each $i
\in \{1,\dots,n\}$, such that there is a sequence of asynchronous
transitions $S_1 \arro{\gamma_1} \dots \arro{\gamma_n}
S_{n+1}$ and such that $S_{n+1}$ has no outgoing
transitions of the same form.
%
%
%
A weak (synchronous/asynchronous) trace of a projected system
$S_1$ is obtained by removing all labels $\co{o^*}:\p a$ and
$\trans{o^*}{a}{b}$ from a strong
(synchronous/asynchronous) trace of $S_1$.
%
\end{definition}

In the definition of traces, input actions and messages
are never considered, since they represent interactions with the
external world, while we are interested in the behavior of closed
systems. 


We have shown in~\cite{SEFM08} that a connected choreography and the
corresponding projected system have the same set of strong traces, if
the synchronous semantics is used for the projected system. The
correspondence with the asynchronous semantics is more complex, and we
refer to~\cite{SEFM08} for a precise description.

\section{Projection}\label{sec:proj}
In this section we show how to derive from a choreography $C$ a
projected system $S$ implementing it. The idea is to project the
choreography $C$ on the different roles, and build the system $S$ as
parallel composition of the projections on the different roles.  We
consider here the most natural projection, which is essentially an
homomorphism on most operators. As shown in~\cite{SEFM08}, if $C$
satisfies the connectedness conditions, the resulting projected system
is behaviorally related to the starting choreography $C$.

\begin{definition}[Projection function]
Given a choreography $C$ and a role $\p a$, the projection $\proj(C,\p a)$
of choreography $C$ on role $\p a$ is defined by structural induction on
$C$:\\

$\begin{array}{rcl}
\proj(\trans{o^?}{a}{b},\p a) & = & \co{o^?}\\
\proj(\trans{o^?}{a}{b},\p b) & = & o^?\\
\proj(\trans{o^?}{a}{b},\p c) & = & \one \textrm{ if } \p c \neq \p a,\p b\\ 
\proj(\one,\p a) & = & \one\\
\proj(\zero,\p a) & = & \zero\\
\proj(C;C',\p a) & = & \proj(C,\p a);\proj(C',\p a)\\
\proj(C \parallel C',\p a) & = & \proj(C,\p a) \mid \proj(C',\p a)\\
\proj(C+C',\p a) & = & \proj(C,\p a)+\proj(C',\p a)\\ 
\end{array}$
\end{definition}

We denote with $\parallel_{i \in I} S_i$ the parallel
composition of systems $S_i$ for each $i \in I$.

\begin{definition}
Given a choreography $C$, the projection of $C$ is the
system $S$ defined by:
$$\proj(C)=\parallel_{\p a \in \roles(C)} \role{\proj(C,\p a)}{a}$$
where $\roles(C)$ is the set of roles in $C$.
\end{definition}

\end{document}